\RequirePackage{fix-cm}
\documentclass[smallextended,twocolumn]{svjour2}       
\smartqed  
\usepackage[english]{babel}
\usepackage{graphicx}

\usepackage{amsmath}
\usepackage{amssymb}
\usepackage{amsfonts}
\usepackage{graphicx}
\usepackage{times}
\usepackage{mathptmx} 
\usepackage{datetime}

\usepackage{tikz}
\usetikzlibrary{arrows}

\usepackage{datetime}
\usepackage{times} 
\usepackage{framed} 
\usepackage{graphicx} 
\usepackage{amsmath} 
\usepackage{amssymb}  

\DeclareMathAlphabet{\bit}{OML}{cmm}{b}{it}

\def\bTheta{\bit{\Theta}}

\newcommand{\ad}{\mathrm{ad}}           
\def\<{\leqslant}           
\def\>{\geqslant}           

\def\d{\partial}

\def\wt{\widetilde}

\def\Re{{\rm Re}}   
\def\rprod{\mathop{\overrightarrow{\prod}}}
\def\lprod{\mathop{\overleftarrow{\prod}}}

\def\mR{{\mathbb R}}    
\def\mC{\mathbb{C}}    

\def\Tr{\mathrm{Tr}}       
\def\rT{\mathrm{T}}        
\def\diam{\diamond}       

\def\[[[{[\![\![}
\def\]]]{]\!]\!]}

\def\re{{\rm e}}        
\def\rd{{\rm d}}        

\def\bEq{\mathbf{E}_{\rm q}}

\def\bEc{\mathbf{E}_{\rm c}}

\def\bJ{\mathbf{J}}

\def\br{{\bf r}}
\def\x{\times}
\def\ox{\otimes}

\def\fA{\mathfrak{A}}
\def\fB{\mathfrak{B}}
\def\fC{\mathfrak{C}}

\def\fH{\mathfrak{H}}

\def\cZ{\mathcal{Z}}

\def\cW{\mathcal{W}}

\def\cX{\mathcal{X}}

\def\cI{\mathcal{I}}

\def\cE{\mathcal{ E}}

\def\eps{\epsilon}

\def\Ups{\Upsilon}

\def\diag{\mathop{\mathrm{diag}}}    
\def\blockdiag{\mathop{\mathrm{blockdiag}}}    

\journalname{}

\begin{document}
\title{Parametric randomization,  complex symplectic factorizations, and quadratic-exponential functionals for Gaussian quantum states$^*$
\thanks{\noindent
$^*$Support:
the Air Force Office of Scientific Research (AFOSR) under agreement number FA2386-16-1-4065 and the Australian
Research Council under grant DP180101805.
Affiliation: Research School of Engineering, College of Engineering and Computer Science, Australian National University, Canberra, Acton, ACT 2601, Australia.
}
}

\titlerunning{Parameter randomization and quadratic-exponential functionals in Gaussian quantum states}        

\author{Igor~G.~Vladimirov\and
Ian~R.~Petersen \and
Matthew~R.~James
}

\authorrunning{Igor~G.~Vladimirov,
Ian~R.~Petersen,
Matthew~R.~James
} 

\institute{Igor G. Vladimirov \at
              \email{igor.g.vladimirov@gmail.com}           
           \and
           Ian R. Petersen \at
           \email{i.r.petersen@gmail.com}
           \and
           Matthew R. James \at
           \email{matthew.james@anu.edu.au}
}

\date{}

\maketitle

\begin{abstract}
This paper combines
probabilistic and
algebraic techniques for computing quantum expectations of operator exponentials (and their products) of quadratic forms of quantum variables in Gaussian states. Such qu\-ad\-ra\-tic-exponential functionals (QEFs) resemble quantum statistical mechanical  partition functions with qu\-ad\-ra\-tic Hamiltonians and are also used as performance criteria in quantum risk-sensitive filtering and control problems for linear quantum stochastic systems. We employ a Lie-algebraic correspondence between complex symplectic matrices  and quadratic-exponential functions of system variables of a quantum harmonic oscillator. The complex symplectic factorizations are used together with a parametric randomization of the quasi-characteristic or moment-generating functions
according to an auxiliary classical Gaussian distribution. This reduces the QEF to an exponential moment of a quadratic form of classical Gaussian random variables with a complex symmetric matrix and is applicable to recursive computation of such moments.
\end{abstract}

\section{Introduction}
\label{intro}

Quantum probabilistic analysis of physical systems, governed by the laws of quantum mechanics, is often concerned with expectations of nonlinear functions of system variables, which are organised as linear operators on an underlying Hilbert space
\cite{M_1995}. When the quantum variables play the role of infinitesimal generators of semigroups, the  generalised moments  involve operator exponentials. For example, in Schwinger's theory of angular momentum \cite{S_1994}, this gives rise to exponentials of bilinear forms of the quantum-mechanical position and momentum operators (with the matrices of these forms being specified by the Levi-Civita symbols from the cross-product of vectors).  Quadratic-exponential functions of quantum variables are also  present in the Gibbs-Boltzmann equilibrium density operators of quantum statistical mechanics \cite{BB_2010} for systems with quadratic Hamiltonians. 

Such systems include quantum harmonic oscillators (QHOs) and their open versions (OQHOs) \cite{GZ_2004} interacting with the environment,
which can involve
quantum fields (for example, nonclassical light), other quantum systems or classical measuring devices, controllers and actuators.
These systems  and their interconnections are usually described in terms of linear quantum stochastic differential equations (QSDEs) and other tools of the Hudson-Parthasarathy  quantum stochastic calculus \cite{H_1991,HP_1984,P_1992,P_2015}.   In combination with quantum feedback network theory \cite{GJ_2009,JG_2010}, this approach provides an important modelling paradigm in linear quantum systems theory \cite{P_2017,WM_2010,ZJ_2012}. The latter is concerned with quantum filtering and control problems \cite{B_1983,B_2010,BH_2006,BVJ_2007,EB_2005,J_2004,J_2005,JNP_2008,NJP_2009,WM_2010} aiming to achieve stability, optimality, robustness and other control objectives.

Quadratic-exponential moments of system variables provide variational performance criteria for open quantum control systems in such problems in the form of
quantum risk-sensitive cost functionals \cite{J_2004,J_2005}. Their minimization by an appropriate choice of feedback control secures a more conservative behaviour of the closed-loop system. This manifests itself in moderation of large deviations of the quantum trajectories \cite{VPJ_2018a} and robustness with respect to
system-environment state uncertainties described in terms of quantum relative entropy \cite{OP_1993,VPJ_2018b,YB_2009}. However, unlike the mean square cost functionals of classical linear-quadratic-Gaussian (LQG) cont\-rol\-lers \cite{KS_1972} and Kalman filters \cite{AM_1979}, whose main features are inherited by their counterparts in measurement-based quantum control and filtering problems, risk-sensitive performance criteria are not additive over time. Rather, they have  a multiplicative structure due to the presence of time-ordered exponentials or operator exponentials of integral quantities.
In view of the noncommutative nature of the quantum system variables (in contrast to classical risk-sensitive control \cite{BV_1985,J_1973,W_1981}) this poses a number of challenging problems in regard to the computation and minimization of such cost functionals. At the same time, it is the exponential penalization of the system variables that makes the quadratic-exponential moments advantageous (compared to the quadratic performance criteria)   in the sense of securing the robustness properties. These issues  are relevant to quantum computing \cite{NC_2000}, quantum optics \cite{WM_2008} and other applications, where controlled isolation and disturbance attenuation for quantum systems are of practical importance.

The present paper is concerned with the computation of quadratic-exponential moments for a finite number of quantum variables satisfying the Weyl canonical commutation relations (CCRs) \cite{F_1989}. This class of operators includes the quantum mechanical positions and momenta (and related annihilation and creation operators \cite{LL_1991,M_1998}) as system variables of a QHO.  We apply a combination of algebraic
and probabilistic techniques to computing the quantum expectations of operator exponentials (and their products) of quadratic forms of the self-adjoint quantum variables in Gaussian states \cite{KRP_2010,P_2015}. As mentioned above, such quadratic-exponential fu\-n\-ctionals (QEFs) resemble quantum statistical mechanical  partition functions
and are used as performance criteria in quantum risk-sensitive filtering and control problems.
We employ an isomorphism between the Lie algebras of complex Hamiltonian matrices and quadratic forms of the quantum variables specified by complex symmetric matrices. Due to Dynkin's representation \cite{D_1947} for the products of exponentials, this leads to a Lie-al\-geb\-ra\-ic correspondence between complex symplectic matrices  and quadratic-expo\-nen\-tial functions of the quantum variables. These complex symplectic factorizations are used together with a double averaging technique whi\-ch involves a parametric randomization of the quasi-cha\-racte\-ristic or moment-generating functions (MGFs)
according to an auxiliary classical Gaussian distribution. This reduces the QEF to an exponential moment of a quadratic form of classical Gaussian random variables with a complex symmetric matrix, which  can be used for recursive computation of such moments and continuous-time extensions of these techniques.
The approach of this paper is related to but differs from the factorization and averaging of quadratic-exponential functions of annihilation and creation operators considered in \cite{CG_2010} in the context of Schwinger's theory and in \cite[Section~3]{PS_2015} in regard to the MGFs for the number operators  in multi-mode QHOs. On a conceptual level, these results have common Lie-algebraic and quantum probabilistic roots (see \cite[Comment~3]{G_1974} and references therein).

The paper is organised as follows.
Section~\ref{sec:start} describes a double averaging technique for parameter-dependent quantum variables. This technique is applied in Section~\ref{sec:CCR} to quasi-characteristic functions of several quantum variables satisfying the Weyl CCRs.
Section~\ref{sec:multi} employs he multivariate parameter randomization for averaging the products of quadratic-exponential functions of such variables.
Section~\ref{sec:QEF} combines
this parameter randomization with the complex symplectic factorizations of Appendix \ref{sec:fact} for computing
the quadratic-exponential moment of Gaussian quantum variables.
Section~\ref{sec:two} outlines recurrence relations for quadratic-exponential functions of qu\-an\-tum variables.
Section~\ref{sec:conc} makes some concluding remarks.
Appendices \ref{sec:quadcom}, \ref{sec:prod}, \ref{sec:fact}  provide auxiliary lemmas on the commutator of quadratic forms of quantum variables satisfying CCRs, their Lie-algebraic isomorphism to complex Hamiltonian matrices, a product formula for quad\-ra\-tic-exponential functions of such variables based on Dynkin's representation, and complex symplectic factorizations for matrices of order two.

\section{A quantum-classical characteristic function identity}\label{sec:start}

As a starting point for what follows, we will describe an identity which involves double averaging.
Let $\xi$ be a self-adjoint quantum variable on (a dense domain of) a complex separable Hilbert space $\fH$ with the spectral decomposition
\begin{equation}
\label{xiS}
    \xi = \int_\mR x S(\rd x),
\end{equation}
where $S$ is a projection-valued measure on the $\sigma$-algebra $\fB$ of Borel subsets of the real line $\mR$ satisfying $S(A)S(B)\\ = S(A\bigcap B)$ for all $A,B\in \fB$, with $S(\mR) = \cI$ the identity operator on $\fH$.
Suppose $\fH$ is endowed with a density operator (quantum state) $\rho$ which is a positive semi-definite self-adjoint operator on $\fH$ with unit trace $\Tr \rho = 1$. Then for any complex-valued function $f$ on $\mR$ (or an interval containing the spectrum of $\xi$), the quantum expectation of the operator
$
    f(\xi) := \int_\mR f(x) S(\rd x)
$
over $\rho$ is given by
\begin{equation}
\label{Efxi}
    \bEq f(\xi)
    :=
    \Tr(\rho f(\xi))
    =
    \int_\mR f(x) D(\rd x).
\end{equation}
Here, $D$ is a probability measure related by
$D(A):= \bEq S(A)$ to the spectral measure $S$ in (\ref{xiS}), so that the self-adjoint operator $\xi$ alone can be regarded as a classical real-valued random variable with the distribution $D$ (see, for example, \cite{H_2001}). In particular, the characteristic function of $\xi$ takes the form
\begin{equation}
\label{cfquant}
    \Phi(u)
    :=
    \bEq \re^{iu \xi}
    =
    \int_\mR \re^{iu x}D(\rd x),
    \qquad
    u \in \mR,
\end{equation}
where $i:= \sqrt{-1}$ is the imaginary unit. In what follows, we will consider the averaging of (\ref{cfquant}) over the parameter $u$, which can be carried out by using
an independent classical real-valued random variable $\omega$ with a probability distribution $F$ and the characteristic function
\begin{equation}
\label{cfclass}
    \varphi(v)
    :=
    \bEc \re^{iv\omega}
    =
    \int_\mR
    \re^{iuv}
    F(\rd u),
    \qquad
    v \in \mR.
\end{equation}
Here, $\bEc(\cdot)$ denotes the classical expectation.\footnote{which can also be considered in the framework of an appropriately augmented quantum probability space $(\fH,\fA,\rho)$, where $\fA$ is a suitable $^*$-algebra of linear operators on $\fH$} The following lemma employs the commutativity between the quantum and classical averaging of parameter-dependent operators:
\begin{align}
\nonumber
    \bEq \bEc g(\omega, \xi)
    & =
    \bEc \bEq g(\omega, \xi)\\
\label{EEcomm}
    & =
    \int_{\mR^2}
    g(u,x)
    F(\rd u) D(\rd x),
\end{align}
where $g: \mR^2\to \mC$ is a complex-valued function on the plane, absolutely integrable over the probability product-measure $F\x D$, and use is made of Fubini's theorem together with (\ref{Efxi}) applied to the function $f(x) := \bEc g(\omega, x) \\ =  \int_\mR g(u,x) F(\rd u) $, and the relation $\bEq g(u, \xi) = \int_\mR g(u,x)\\ D(\rd x)$.

\begin{lemma}
\label{lem:cfuni}
The characteristic functions $\Phi$ and $\varphi$ of a self-adjoint quantum variable $\xi$ and an independent  classical real-valued random variable $\omega$ in (\ref{cfquant}) and (\ref{cfclass}), respectively, are related by
\begin{equation}
\label{EPhi}
    \bEq\varphi(\xi)
    =
    \bEc \Phi(\omega).
\end{equation}
\end{lemma}
\begin{proof}
The relation (\ref{EPhi}) is obtained by applying (\ref{EEcomm}) to the function $g(u,x):= \re^{iux}$ and noting that $\bEq \re^{iu \xi} = \Phi(u)$ and $\bEc \re^{ix \omega} = \varphi(x)$ for all $u,x\in \mR$ in view of (\ref{cfquant}) and (\ref{cfclass}).
\end{proof}

The left-hand side of (\ref{EPhi}) is a generalised moment of the quantum variable $\xi$, whereas its right-hand side is a generalised moment of the classical random variable $\omega$. These two moments are specified by the characteristic functions $\varphi$ and  $\Phi$ and  the probability distributions $D$ and $F$.

In particular, if $\omega$ is a Gaussian random variable with zero mean and variance $\sigma^2$, then its characteristic function $\varphi$ in (\ref{cfclass}) takes the form
\begin{equation}
\label{cfclassgauss}
    \varphi(v) = \re^{-\frac{1}{2}\sigma^2 v^2},
\end{equation}
and substitution into (\ref{EPhi}) leads to
\begin{align}
\nonumber
    \bEq\re^{-\frac{1}{2} \sigma^2 \xi^2}
    & =
    \bEc \Phi(\omega)\\
\label{EPhiGauss}
    & =
    \frac{1}{\sqrt{2\pi} \sigma}
    \int_\mR
    \Phi(u)
    \re^{-\frac{u^2}{2\sigma^2}}
    \rd u,
\end{align}
where the integral applies to the case $\sigma >0$.
The left-hand side of (\ref{EPhiGauss}) is a quadratic-exponential moment of the quantum variable $\xi$ which is reduced on the right-hand side to a classical average of its characteristic function over a Gaussian distribution for the argument of the function. For example, if $\xi$ is in a zero-mean  Gaussian quantum state with a variance $\Sigma^2$, then, similarly to (\ref{cfclassgauss}), the characteristic function $\Phi$ in (\ref{cfquant}) takes the form
\begin{equation}
\label{cfquantgauss}
    \Phi(u) = \re^{-\frac{1}{2}\Sigma^2 u^2},
\end{equation}
which, in combination with (\ref{EPhiGauss}), yields
\begin{align}
\nonumber
    \bEq\re^{-\frac{1}{2} \sigma^2 \xi^2}
    & =
    \bEc \re^{-\frac{1}{2}\Sigma^2 \omega^2}\\
\nonumber
    & =
    \frac{1}{\sqrt{2\pi} \sigma}
    \int_\mR
    \re^{-\frac{1}{2}\big(\Sigma^2 + \frac{1}{\sigma^2}\big)u^2}
    \rd u\\
\label{EPhiGaussGauss}
    & =
    \frac{1}{\sqrt{1 + \Sigma^2 \sigma^2}}.
\end{align}
The quantum-classical characteristic function identity (\ref{EPhi}) and its Gaussian version in (\ref{cfclassgauss})--(\ref{EPhiGaussGauss}) are valid irrespective of the particular structure of the quantum variable $\xi$ or the Hilbert space $\fH$. At the same time, the case of a single self-adjoint operator considered above does not contain, in itself, much ``quantumness'' because of the isomorphism with a classical random variable. Nevertheless, the main idea of these identities (using the parameter randomization) admits an extension to a more general setting of several (and even infinitely many) noncommuting quantum variables without a joint probability distribution. More precisely, let $\xi_1, \ldots, \xi_n$ be self-ad\-joint quantum variables (on a dense domain of the underlying Hilbert space $\fH$), and let  $\omega_1, \ldots, \omega_n$ be independent classical real-valued random variables. Then \begin{align}
\nonumber
    \bEc &
    \bEq
    \rprod_{k=1}^n
    g_k(\omega_k,\xi_k)\\
\nonumber
    & =
    \int_{\mR^n}
    \bEq
    \rprod_{k=1}^n
    g_k(u_k,\xi_k)
    F_1(\rd u_1)\x \ldots \x F_n(\rd u_n)\\
\nonumber
    & =
    \int_{\mR^n}
    \Tr
    \Big(
    \rho
    \rprod_{k=1}^n
    g_k(u_k,\xi_k)
    \Big)
    F_1(\rd u_1)\x \ldots \x F_n(\rd u_n)\\
\nonumber
    & =
    \Tr
    \Big(
    \rho
    \rprod_{k=1}^n
    \int_\mR
    g_k(u_k,\xi_k)
    F_k(\rd u_k)
    \Big)\\
\label{EEEcomm}
    & =
    \bEq
    \rprod_{k=1}^n
    \bEc
    g_k(\omega_k,\xi_k).
\end{align}
Here, $g_1, \ldots, g_n:\mR^2 \to \mC$ are Borel measurable complex-valued functions on the plane, and $F_1, \ldots, F_n$ are the mar\-gi\-nal probability distributions  of $\omega_1, \ldots, \omega_n$. Also,
\begin{equation}
\label{xirprod}
    \rprod_{k=1}^n
    \zeta_k
    :=
    \zeta_1\x \ldots  \x \zeta_n
\end{equation}
denotes
the ``rightward'' product of elements $\zeta_1, \ldots, \zeta_n$ in a noncommutative ring (for example, quantum variables or complex matrices). The identity  (\ref{EEEcomm}) is valid at least in the case when the functions $g_1, \ldots, g_n$ are bounded.
Since  the quantum variables $\xi_1, \ldots, \xi_n$ do not commute with each other, in general, they do not have a classical joint  probability distribution on $\mR^n$, and the quantum expectations in (\ref{EEEcomm}) are not reducible to such a distribution (in contrast to the case of a single self-adjoint quantum variable in (\ref{Efxi})).

\section{Unitary Weyl operators and quasi-characteristic functions}\label{sec:CCR}

Consider $n$ self-adjoint quantum variables $X_1, \ldots, X_n$ on (a dense domain of) the underlying Hilbert space $\fH$, satisfying the Weyl  CCRs
\begin{equation}
\label{CCR}
    \cW_u \cW_v
    =
    \re^{i v^{\rT}\Theta u}
    \cW_{u+v},
    \qquad
    u,v\in \mR^n,
\end{equation}
where
$\Theta:= (\theta_{jk})_{1\< j,k\< n}$ is a real antisymmetric matrix of order $n$. Here, associated with the vector
\begin{equation}
\label{X}
    X:=
    \begin{bmatrix}
        X_1\\
        \vdots\\
        X_n
    \end{bmatrix}
\end{equation}
of the quantum variables (unless indicated otherwise, vectors are organised  as columns) are the unitary Weyl operators \cite{F_1989}
\begin{equation}
\label{cW}
  \cW_u
  :=
  \re^{iu^{\rT} X}
\end{equation}
on $\fH$
parameterised by vectors $u:=(u_k)_{1\< k\< n} \in \mR^n$.
The self-adjointness of the operator $u^{\rT}X = \sum_{k=1}^n u_k X_k$ implies the relation
\begin{equation}
\label{cWdagger}
    \cW_u^{\dagger} = \cW_{-u}
\end{equation}
for the adjoint of the Weyl operator in (\ref{cW}). It follows from (\ref{CCR}) that
\begin{equation}
\label{cWcom}
    [\cW_u, \cW_v]
    =
    -2i\sin(u^{\rT}\Theta v)
    \cW_{u+v},
\end{equation}
where $[\xi,\eta]:=\xi\eta - \eta \xi$ is the commutator of linear operators. When considered at the level of infinitesimal generators (for small values of $u$ and $v$), the relation
(\ref{cWcom}) yields
the Heisenberg  form of the Weyl CCRs (\ref{CCR}) described by the commutator matrix
\begin{align}
\nonumber
    [X, X^{\rT}]
      & :=    ([X_j,X_k])_{1\< j,k\< n}\\
\label{Theta}
     & =    XX^{\rT}- (XX^{\rT})^{\rT}
     =
     2i \Theta,
\end{align}
where the transpose $(\cdot)^{\rT}$ applies to a matrix with operator-valued entries as if they were scalars.
The entries $2i \theta_{jk}$ of the matrix on the right-hand side of (\ref{Theta}) represent the operators $2i \theta_{jk}\cI$.
By a standard convention on linear operators, the corresponding matrix $\Theta \ox \cI$ (where $\ox$ is the tensor product) is identified with $\Theta$. The Weyl CCRs (\ref{CCR}) can also be regarded as a particular case of the Baker-Campbell-Hausdorff formula \index{Baker-Campbell-Hausdorff formula} $\re^{\xi+\eta} = \re^{-\frac{1}{2}[\xi,\eta]}\re^{\xi} \re^{\eta}$ for operators $\xi$ and $\eta$ which commute with their commutator, that is, $[\xi,[\xi,\eta]] = [\eta,[\xi,\eta]] = 0$ (see, for example, \cite[pp. 128--129]{GZ_2004}). Indeed, (\ref{Theta}) leads to  $[iu^{\rT} X, \\ iv^{\rT} X] = -u^{\rT}[X,X^{\rT}] v = -2iu^{\rT}\Theta v$ which commutes with any operator.

The quantum expectation of the Weyl operator $\cW_u$ in (\ref{cW}) gives rise to the complex-valued quasi-characteristic function (QCF)
\begin{equation}
\label{Phi}
    \Phi(u):= \bEq \cW_u = \overline{\Phi(-u)},
    \qquad
    u \in \mR^n,
\end{equation}
where the second equality  describes the Hermitian property of $\Phi$ and follows from (\ref{cWdagger}). The relation $\cW_0 = \cI$ implies that $\Phi(0)=1$. Furthermore, the unitarity of $\cW_u$ leads to $|\Phi(u)|\< 1$ for all $u \in \mR^n$, similarly to the classical case. However, in view of the Weyl CCRs (\ref{CCR}) in the quantum mechanical setting, the Bochner-Khinchin positiveness criterion \cite{GS_2004}  for the characteristic functions of classical probability distributions is replaced with its modified version \cite{CH_1971,H_2010} in the form of  positive semi-definiteness of the complex Hermitian matrix  $\big(\re^{iu_j^{\rT}\Theta u_k}\\  \Phi(u_j-u_k)\big)_{1\<j,k\<N}$ for any $u_1, \ldots, u_N \in \mR^n$ and $N=1,2, 3,\ldots$.
In the case $\Theta\ne 0$,  the QCF $\Phi$ itself does not have to be positive \cite{H_1974}.

In order to extend the parameter randomization technique of Section~\ref{sec:start} to the multivariable case, we will need a factorization of the Weyl operator $\cW_u$ in (\ref{cW}) into univariate Weyl operators, which is given below for completeness (see also \cite[Appendix B]{VP_2012a} or \cite[Section 6]{V_2015c}). To this end, for any matrix $M:= (m_{jk})_{1\< j,k\< n}$, we denote by $M^\diam$ the  matrix of the same order with the entries
\begin{equation}
\label{diam}
    (M^\diam)_{jk}
    :=
    \left\{
        \begin{matrix}
            m_{jk} & & {\rm if} & j \< k\\
            m_{kj} & & {\rm if} & j > k
        \end{matrix}
    \right.,
\end{equation}
so that $M^\diam$ is a symmetric matrix which 
inherits its main diagonal and the upper triangular part (above the main diagonal) from $M$. Only symmetric matrices are invariant under the linear operation $(\cdot)^\diam$ (that is, $M^\diam = M$ is equivalent to $M^{\rT}=M$). Also, if $M$ is antisymmetric, then $M^{\diam}$ has zeros over the main diagonal.

\begin{lemma}
\label{lem:cWfact}
The Weyl operator  (\ref{cW}), associated with the self-adjoint quantum variables $X_1, \ldots, X_n$ satisfying the Weyl CCRs (\ref{CCR}), admits the factorization
\begin{equation}
\label{Baker}
    \cW_u
    =
    \re^{\frac{i}{2}u^{\rT} \Theta^{\diam} u}
    \re^{iu_1X_1}
    \x \ldots \x
    \re^{iu_nX_n}
\end{equation}
for any $u:=(u_k)_{1\< k\< n}\in \mR^n$.  Here, $\Theta^\diam$ is a real symmetric matrix of order $n$ defined by (\ref{diam}) in terms of the CCR matrix $\Theta$ in (\ref{Theta}).
\end{lemma}
\begin{proof}
The factorization (\ref{Baker}) can be obtained by repeatedly applying the Weyl CCRs (\ref{CCR}), from which it follows by induction  that
\begin{align}
\nonumber
    \cW_u
      & =
      \re^{i\sum_{j=1}^{n-1}u_j e_j^{\rT}\Theta e_n u_n }
      \cW_{\sum_{j=1}^{n-1}u_j e_j}
      \cW_{u_n e_n}\\
\nonumber
      &=
    \rprod_{k=1}^n
    \re^{i\sum_{j=1}^{k-1}u_j e_j^{\rT}\Theta e_ku_k }
    \cW_{u_k e_k}\\
\nonumber
        & =
    \re^{i\sum_{1\< j< k \< n}\theta_{jk} u_ju_k}
    \rprod_{k=1}^n
    \cW_{u_ke_k}\\
\label{Baker1}
    & =
    \re^{\frac{i}{2}u^{\rT} \Theta^\diam u}
    \rprod_{k=1}^n
    \cW_{u_ke_k}
\end{align}
for any $u:=(u_k)_{1\< k\< n}\in \mR^n$, where
$$
    e_k:= [\underbrace{0,\ldots, 0}_{k-1\ {\rm zeros}}, 1, \underbrace{0,\ldots, 0}_{n-k\ {\rm zeros}}]^{\rT}
$$
is  the $k$th standard basis vector in $\mR^n$, so that $\cW_{u_k e_k} = \re^{iu_k X_k}$, and (\ref{Baker1}) establishes (\ref{Baker}). Here, use is made of (\ref{xirprod}) and the identity
$
    \sum_{1\< j< k \< n}
    \theta_{jk}
    u_ju_k
    =
    \frac{1}{2}
    u^{\rT}\Theta^\diam u
$,
where, in view of (\ref{diam}),
the real symmetric matrix $\Theta^\diam$
inherits the upper off-diagonal entries and the zero main diagonal from $\Theta$.
\end{proof}

By taking the quantum expectation on both sides of (\ref{Baker}) and recalling the QCF $\Phi$ from (\ref{Phi}), it follows that
\begin{equation}
\label{Phiprod}
    \bEq
    \rprod_{k=1}^n
    \re^{iu_k X_k}
    =
  \Phi(u)
    \re^{-\frac{i}{2}u^{\rT} \Theta^\diam u},
    \qquad
    u \in \mR^n.
\end{equation}
Although the left-hand side of (\ref{Phiprod}) is a special parameter-dependent moment of the quantum variables $X_1, \ldots, X_n$, this relation can be used for computing more general moments.

\section{Parameter randomization in the multivariable case}\label{sec:multi}

As in (\ref{EEEcomm}), let   $\omega_1, \ldots, \omega_n$ be independent classical real-valued random variables with characteristic functions $\varphi_1,\\ \ldots, \varphi_n$, respectively. They are assembled into a random vector   \begin{equation}
\label{omega}
    \omega
    :=
    \begin{bmatrix}
        \omega_1\\
        \vdots\\
        \omega_n
    \end{bmatrix}
\end{equation}
whose characteristic function is the tensor product of the marginal characteristic functions: $(\varphi_1 \ox \ldots \ox \varphi_n)(u) = \varphi_1(u_1)\x \ldots \x \varphi_n(u_n)$ for any $u:=(u_k)_{1\< k\< n}\in \mR^n$.
\begin{theorem}
\label{th:cfmulti}
For self-adjoint quantum variables $X_1, \ldots, \\ X_n$, which satisfy the Weyl CCRs (\ref{CCR}) and have the QCF $\Phi$ in (\ref{Phi}), and independent classical real-valued random variables $\omega_1, \ldots, \omega_n$ with characteristic functions $\varphi_1, \ldots, \varphi_n$, these functions are related by
\begin{equation}
\label{Ephiprod}
    \bEq
    \rprod_{k=1}^n
    \varphi_k(X_k)
    =
    \bEc
    \Big(
  \Phi(\omega)
    \re^{-\frac{i}{2}\omega^{\rT} \Theta^{\diam} \omega}
    \Big).
\end{equation}
Here, $\omega$ is the classical random vector from (\ref{omega}), and, as in Lemma~\ref{lem:cWfact}, the matrix $\Theta^\diam$ is associated by (\ref{diam}) with the CCR matrix $\Theta$ in (\ref{Theta}).
\end{theorem}
\begin{proof}
By evaluating both parts of (\ref{Phiprod}) at the random vector $u =\omega$ and applying the classical expectation, it follows that
\begin{equation}
\label{EE}
    \bEc\bEq
    \rprod_{k=1}^n
    \re^{i\omega_k X_k}
    =
    \bEc
    \big(
  \Phi(\omega)
    \re^{-\frac{i}{2}\omega^{\rT} \Theta^\diam \omega}
    \big).
\end{equation}
Since the classical and quantum expectations commute with each other, and the random variables $\omega_1, \ldots, \omega_n$ are independent, the left-hand side of (\ref{EE}) can be computed by using (\ref{EEEcomm}) as
\begin{equation}
\label{EEleft}
    \bEc\bEq
    \rprod_{k=1}^n
    \re^{i\omega_k X_k}
    =
    \bEq
    \rprod_{k=1}^n
    \bEc
    \re^{i\omega_k X_k}
    =
    \bEq
    \rprod_{k=1}^n
    \varphi_k(X_k).
\end{equation}
This is a multivariable counterpart of (\ref{EPhi}), which employs the fact that  $\re^{i\omega_1 X_1}, \ldots,\re^{i\omega_n X_n}$ are statistically independent random elements (random unitary operators on the Hilbert space $\fH$). A combination of (\ref{EE}) and (\ref{EEleft}) leads to (\ref{Ephiprod}).
\end{proof}

In particular, suppose $\omega_1, \ldots, \omega_n$ are independent standard normal (zero-mean unit-variance Gaussian) random variables. Then their characteristic  functions $\varphi_1, \ldots, \varphi_n$ are given by (\ref{cfclassgauss}) with $\sigma = 1$, and the relation (\ref{Ephiprod}) takes the form
\begin{equation}
\label{Ephiprodgauss}
    \bEq
    Y
    =
    \frac{1}{(2\pi)^{n/2}}
    \int_{\mR^n}
  \Phi(u)
    \re^{-\frac{1}{2}u^{\rT} (I_n+i\Theta^\diam)u}
    \rd u
\end{equation}
(with $I_n$ the identity matrix of order $n$)
whose left-hand side is a quadratic-exponential product moment of the quantum variables $X_1, \ldots, X_n$ involving the operator
\begin{equation}
\label{Y}
    Y
    :=
    \rprod_{k=1}^n
    \re^{-\frac{1}{2} X_k^2},
\end{equation}
which is a contraction due to positive semi-definiteness of the operators $X_1^2, \ldots, X_n^2$.
The right-hand side of (\ref{Ephiprodgauss}) admits closed-form calculation, for example, if $X_1, \ldots, X_n$ are in a zero-mean Gaussian state \cite{KRP_2010} with a quantum covariance matrix
\begin{equation}
\label{covX}
    \bEq (XX^{\rT})
    =
    P+i\Theta \succcurlyeq 0
\end{equation}
(whose positive semi-definiteness reflects the generalized Heisenberg uncertainty principle \cite{H_2001}),
where $P$ is a real positive semi-definite symmetric matrix of order $n$. In this case, the QCF (\ref{Phi}) is given by
\begin{equation}
\label{Phigauss}
    \Phi(u)
    =
  \re^{-\frac{1}{2}\|u\|_P^2},
  \qquad
  u \in \mR^n,
\end{equation}
where $\|u\|_P:= \sqrt{u^{\rT} P u} = |\sqrt{P} u|$ is a weighted Euclidean semi-norm. Substitution of (\ref{Phigauss}) into (\ref{Ephiprodgauss}) reduces the quadratic-exponential product moment to a Gaussian integral
\begin{align}
\nonumber
    \bEq
    Y
    & =
    \frac{1}{(2\pi)^{n/2}}
    \int_{\mR^n}
    \re^{-\frac{1}{2}u^{\rT} (P+I_n+i\Theta^\diam)u}
    \rd u\\
\label{gint}
    & =
    \frac{1}{\sqrt{\det(P+I_n+i\Theta^\diam)}}.
\end{align}
In general, the quantity on the right-hand side of (\ref{gint}) is complex-valued, which is in line with the fact that the operator $Y$ in (\ref{Y})
is not self-adjoint. A real-valued version of the quadratic-exponential product moment above can be defined by
\begin{equation}
\label{QYY}
    \bEq
    (YY^{\dagger})
    =
    \bEq
    \left(
    \rprod_{j=1}^n
    \re^{-\frac{1}{2} X_j^2}
    \lprod_{k=1}^n
    \re^{-\frac{1}{2} X_k^2}
    \right),
\end{equation}
where, similarly to (\ref{xirprod}),
    $$\lprod_{k=1}^n
    \zeta_k
    :=
    \zeta_n\x \ldots  \x \zeta_1$$
is the ``leftward'' product of quantum variables $\zeta_1, \ldots, \zeta_n$. The moment (\ref{QYY}) is real-valued since the operator $YY^\dagger$ is self-adjoint. Its computation can be carried out by using an augmented vector of the quantum variables $X_1, \ldots, X_n$ duplicated in reverse order as
\begin{equation}
\label{cX}
    \cX:=
    \begin{bmatrix}
        X_1\\
        \vdots\\
        X_n\\
        X_n\\
        \vdots\\
        X_1
    \end{bmatrix}
    =
    \begin{bmatrix}
        X\\
        RX
    \end{bmatrix}
    =
    \begin{bmatrix}
        I_n\\
        R
    \end{bmatrix}
    X,
\end{equation}
where
\begin{equation}
\label{R}
  R:=
    \begin{bmatrix}
        0 & \cdots & 1\\
        \vdots &\cdot^{\cdot^\cdot} & \vdots \\
        1 & \cdots & 0
    \end{bmatrix}
\end{equation}
is a Hankel matrix of order $n$ whose entries are zeros except for the main anti-diagonal which consists of ones. The vector $\cX$ in (\ref{cX}) is also in a zero-mean Gaussian state, and its quantum covariance matrix is related to that of $X$ in (\ref{covX}) by
\begin{equation}
\label{covcX}
    K
    :=
    \bEq (\cX\cX^{\rT})
    =
    \begin{bmatrix}
        I_n\\
        R
    \end{bmatrix}
    (P+i\Theta)
    \begin{bmatrix}
        I_n & R
    \end{bmatrix},
\end{equation}
where use is made of the symmetry of $R$ in (\ref{R}). Therefore, application of Theorem~\ref{th:cfmulti} to the augmented vector $\cX$ in (\ref{cX}) yields an appropriate modification of (\ref{gint}):
\begin{align}
\nonumber
    \bEq (Y&Y^{\dagger})
    =
    \bEq
    \Big(
    \rprod_{j=1}^n
    \varphi(X_j)
    \lprod_{k=1}^n
    \varphi(X_k)
    \Big)\\
\nonumber
    & =
    \bEc
    \Big(
  \Phi(\omega)
    \exp
    \Big(
    -\frac{i}{2}\omega^{\rT}
    \Big(
        \begin{bmatrix}
        I_n\\
        R
    \end{bmatrix}
    \Theta
    \begin{bmatrix}
        I_n & R
    \end{bmatrix}
    \Big)
    ^{\diam} \omega
    \Big)
    \Big)\\
\nonumber
    & =
    \bEc
    \re^{
    -\frac{1}{2}\omega^{\rT}K
    ^{\diam} \omega
    }\\
\nonumber
    & =
    \frac{1}{(2\pi)^n}
    \int_{\mR^{2n}}
    \re^{-\frac{1}{2}u^{\rT} (I_{2n}+K^\diam)u}
    \rd u\\
\label{gintYY}
    & =
    \frac{1}{\sqrt{\det(I_{2n}+K^\diam)}},
\end{align}
where $\omega$ is a classical zero-mean Gaussian $\mR^{2n}$-valued random vector with the identity covariance matrix (so that $\omega$ consists of $2n$ independent standard normal random variables). The right-hand side of (\ref{gintYY}) is a positive quantity, which does not exceed $1$ since the operator $YY^{\dagger}$ is also a contraction (inheriting this property from $Y$ in (\ref{Y})). Moreover, it admits the upper  bound
\begin{align}
\nonumber
    \bEq (YY^{\dagger})
    & \<
    \frac{1}{(2\pi)^n}
    \int_{\mR^{2n}}
    \re^{-\frac{1}{2}u^{\rT} (I_{2n}+\Re K^\diam)u}
    \rd u\\
\nonumber
    & =
    \frac{1}{\sqrt{\det(I_{2n}+\Re K)}}\\
\label{EYYupper}
    & =
    \frac{1}{\sqrt{\det(I_n+2P)}}.
\end{align}
Here, we have also used the fact that the real part of the quantum covariance matrix $K$ in (\ref{covcX}) is the symmetric matrix
$$    \Re K
    =
    \begin{bmatrix}
        I_n\\
        R
    \end{bmatrix}
    P
    \begin{bmatrix}
        I_n & R
    \end{bmatrix}
$$
whose spectrum differs only by zeros from the spectrum of the matrix $2P$ due to the orthogonality of the symmetric matrix $R$ in (\ref{R}):
$$
    \begin{bmatrix}
        I_n & R
    \end{bmatrix}
    \begin{bmatrix}
        I_n\\
        R
    \end{bmatrix}
    P
    =
    (I_n + R^2) P = 2P.
$$
The real-valuedness of the determinant in (\ref{gintYY}) can also be established directly by computing the determinant  in a recursive fashion.

By using the transformation $\omega \mapsto -i\omega$, which replaces the characteristic functions with the MGFs, and an analytic continuation argument,  (\ref{Ephiprod}) of  Theorem~\ref{th:cfmulti} can be modified as
\begin{equation}
\label{Epsiprod}
    \bEq
    \rprod_{k=1}^n
    \psi_k(X_k)
    =
    \bEc
    \Big(
  \Psi(\omega)
    \re^{\frac{i}{2}\omega^{\rT} \Theta^{\diam} \omega}
    \Big),
\end{equation}
provided the expectations exist.
Here, $\psi_1, \ldots, \psi_n$ are the marginal MGFs of the independent classical random variables $\omega_1, \ldots, \omega_n$, and
\begin{equation}
\label{Psi}
  \Psi(u)
  :=
  \Phi(-iu)
  =
  \bEq \re^{u^\rT X},
  \qquad
  u \in \mR^n,
\end{equation}
is the MGF for the quantum variables $X_1, \ldots, X_n$.

\section{Averaging the quadratic-exponential functions of Gaussian quantum variables}
\label{sec:QEF}

We will now combine the parameter randomization technique of the previous section with complex symplectic factorizations of Appendix~\ref{sec:fact} in order to compute
the QEF
\begin{equation}
\label{Xi}
    \Xi := \bEq \re^{X^\rT \Pi X}
\end{equation}
for the vector $X$ of self-adjoint quantum variables $X_1, \ldots,\\ X_n$ in (\ref{X}), satisfying the Weyl CCRs (\ref{CCR}), where
$\Pi$ is a real positive definite symmetric matrix of order $n$. The following formulates relevant notation and conditions and, in essence, provides  the proof of Theorem~\ref{th:QEF} below.
Let the CCR matrix $\Theta$ be nonsingular. Then there exists a nonsingular matrix $T \in \mR^{n\x n}$ such that
\begin{equation}
\label{TJT}
    \Theta = T J T^\rT,
    \qquad
    J:= \frac{1}{2}(I_\nu \ox \bJ) ,
    \qquad
    \nu := \frac{n}{2},
\end{equation}
where $\ox$ is the Kronecker product of matrices, and
\begin{equation}
\label{bJ}
    \bJ
    :=
    \begin{bmatrix}
        0 & 1\\
        -1 & 0
    \end{bmatrix}
\end{equation}
spans the subspace of antisymmetric matrices of order $2$. Hence,
$T^{-1}X$
consists of self-adjoint quantum variables which satisfy
\begin{align*}
    [T^{-1}X,(T^{-1}X)^\rT]
    & = T^{-1}[X,X^\rT] T^{-\rT}\\
    & = 2i T^{-1}\Theta T^{-\rT} = 2iJ
\end{align*}
with the CCR matrix $J$, where $(\cdot)^{-\rT}:= ((\cdot)^{-1})^\rT$.
Since $T^\rT \Pi T\succ 0$, then, by Williamson's symplectic diagonalization  theorem \cite{W_1936,W_1937} (see also \cite[pp. 244--245]{D_2006}), there exists a symplectic matrix $V \in \mR^{n\x n}$ (satisfying $V JV^\rT= J$) such that
\begin{equation}
\label{TPT}
    V^\rT T^\rT \Pi T V = \Lambda \ox I_2,
    \qquad
    \Lambda
    := \diag_{1\< k\< \nu}(\lambda_k),
\end{equation}
where $\lambda_1, \ldots, \lambda_\nu$ are positive real numbers (the symplectic eigenvalues of the matrix $T^\rT \Pi T$).  Hence,
\begin{equation}
\label{Z}
  Z:= V^{-1}T^{-1} X
\end{equation}
inherits the CCR matrix $J$ from the vector $T^{-1}X$ and,
up to an isomorphism,   can be formed from $\nu$ pairs $(q_k, p_k)$ of conjugate position $q_k$ and momentum $p_k:= -i\d_{q_k}$ operators, $k=1,\ldots, \nu$, acting on the Schwartz space \cite{V_2002} on $\mR^{\nu}$:
\begin{equation}
\label{Zqp}
    Z
    :=
    \begin{bmatrix}
        q_1\\
        p_1\\
        \vdots\\
        q_{\nu}\\
        p_{\nu}
    \end{bmatrix}.
\end{equation}
Indeed, each of the position-momentum pairs has the CCR matrix $\frac{1}{2} \bJ$, and the quantum variables in different pairs commute with each other:
\begin{equation}
\label{qpcomm}
    \left[
    \begin{bmatrix}
      q_j\\
      p_j
    \end{bmatrix},
    \begin{bmatrix}
      q_k\\
      p_k
    \end{bmatrix}^\rT
    \right]
    =
    i \delta_{jk}\bJ,
    \qquad
    j,k=1,\ldots, \nu,
\end{equation}
where $\delta_{jk}$ is the Kronecker delta, and the matrix $\bJ$ is given by (\ref{bJ}). From (\ref{TPT})--(\ref{Zqp}), it follows that
\begin{align*}
    X^\rT \Pi X
    & = Z^\rT  V^\rT T^\rT \Pi TVZ\\
    & = Z^\rT  (\Lambda \ox I_2)Z\\
    & =
    \sum_{k=1}^\nu
    \lambda_k (q_k^2+p_k^2),
\end{align*}
which, in view of the interpair commutativity in (\ref{qpcomm}), implies that
\begin{equation}
\label{eXPX}
    \re^{X^\rT \Pi X }
    =
    \prod_{k=1}^\nu
    \re^{\lambda_k (q_k^2+p_k^2)}.
\end{equation}
Application of (\ref{eee}) and (\ref{sympfact3}) to each of the factors in (\ref{eXPX}) leads to
\begin{equation}
\label{eeek}
  \re^{\lambda_k (q_k^2+p_k^2)}
  =
    \re^{\alpha_k q_k^2}
    \re^{\beta_k p_k^2}
    \re^{\alpha_k q_k^2}
\end{equation}
where
\begin{equation}
\label{sympfactk}
    \alpha_k
    =
    \frac{1}{2}
    \tanh \lambda_k,
    \qquad
   \beta_k
     =
    \frac{1}{2}
    \sinh(2\lambda_k),
    \qquad
    k = 1,\ldots, \nu,
\end{equation}
and hence,
\begin{equation}
\label{eXPX1}
    \re^{X^\rT \Pi X }
    =
    \prod_{k=1}^\nu
    \big(
    \re^{\alpha_k q_k^2}
    \re^{\beta_k p_k^2}
    \re^{\alpha_k q_k^2}
    \big),
\end{equation}
whose right-hand side is the product of $\nu$ pairwise commuting positive definite   self-adjoint quantum variables (each of which is a product of three noncommuting quantum variables).
We will now rearrange the positions and momenta in triples to form a vector $\cZ:= (\cZ_k)_{1\< k \< 3\nu}$ as
\begin{equation}
\label{cZ}
    \cZ
    :=
    \begin{bmatrix}
        q_1\\
        p_1\\
        q_1\\
        \vdots\\
        q_\nu\\
        p_\nu\\
        q_\nu
    \end{bmatrix}
    =
    F Z,
    \qquad
    F:=
    I_\nu
    \ox
    \begin{bmatrix}
      1 & 0\\
      0 & 1\\
      1 & 0
    \end{bmatrix},
\end{equation}
where the structure of the binary matrix $F \in \{0,1\}^{3\nu\x n}$ follows from (\ref{Zqp}). By taking the quantum expectation on both sides of (\ref{eXPX1}) and applying (\ref{Epsiprod}) and (\ref{Psi}) to (\ref{cZ}), it follows that the QEF (\ref{Xi}) can be computed as
\begin{equation}
\label{Epsiprod1}
    \Xi
    =
    \bEq
    \rprod_{k=1}^{3\nu}
    \psi_k(\cZ_k)
    =
    \bEc
    \Big(
  \Psi(\omega)
    \re^{\frac{i}{2}\omega^{\rT} (FJF^\rT)^{\diam} \omega}
    \Big).
\end{equation}
Here, $\omega:= (\omega_k)_{1\< k \< 3\nu}$ is a classical random vector consisting of independent zero-mean Gaussian random variables with variances
\begin{align}
\label{vars1}
    \bEc (\omega_{3k-2}^2) & = \bEc (\omega_{3k}^2) = 2\alpha_k,\\
\label{vars2}
    \bEc (\omega_{3k-1}^2)
    & = 2\beta_k,
    \qquad
    k = 1, \ldots, \nu,
\end{align}
and the corresponding MGFs
\begin{align*}
    \psi_{3k-2}(z)
    & = \psi_{3k}(z) = \re^{\alpha_k z^2},\\
    \psi_{3k-1}(z)
    & = \re^{\beta_k z^2}.
\end{align*}
In (\ref{Epsiprod1}), use is also made of the operation $(\cdot)^\diam$ from (\ref{diam}),   together with  the MGF $\Psi$ and the CCR matrix $FJF^\rT$ of the vector $\cZ$ in (\ref{cZ}), with the matrix $J$ given by (\ref{TJT}). Therefore, if the quantum variables $X_1, \ldots, X_n$ in (\ref{X}) are in a zero-mean Gaussian state with the quantum covariance matrix (\ref{covX}), then $\cZ$ is also in a zero-mean Gaussian state with 
\begin{align}
\nonumber
    L
    & :=
    \bEq(\cZ\cZ^\rT)\\
\nonumber
     & =
    F\bEq(ZZ^\rT)F^\rT\\
\nonumber
    & =
    FV^{-1}T^{-1} (P+i\Theta)T^{-\rT}V^{-\rT}F^\rT\\
\label{covcZ}
    & =
    FV^{-1}T^{-1}P T^{-\rT}V^{-\rT}F^\rT
    +
    iFJF^\rT,
\end{align}
and hence, its MGF takes the form
\begin{equation}
\label{MGFcZ}
    \Psi(u)
    =
    \bEq \re^{u^\rT \cZ}
    =
    \re^{\frac{1}{2}\|T^{-\rT}V^{-\rT}F^\rT u\|_P^2},
    \qquad
    u \in \mR^{3\nu}.
\end{equation}
Substitution of (\ref{MGFcZ}) into (\ref{Epsiprod1}) leads to
\begin{align}
\nonumber
    \Xi
    & =
    \bEc
  \re^{\frac{1}{2}(\|T^{-\rT}V^{-\rT}F^\rT \omega\|_P^2 +
    i\omega^{\rT} (FJF^\rT)^{\diam} \omega)}\\
\nonumber
    & =
    \bEc
  \re^{\frac{1}{2}\omega^\rT L^\diam \omega}\\
\nonumber
    & =
    \frac{1}{(2\pi)^{3\nu/2}\sqrt{\det \mho}}
    \int_{\mR^{3\nu}}
    \re^{-\frac{1}{2}u^{\rT} (\mho^{-1}-L^\diam)u}
    \rd u\\
\nonumber
  & =
  \frac{1}{\sqrt{\det(I_{3\nu} - \mho L^\diam)}}\\
  \label{Epsiprod2}
  & \<
  \frac{1}{\sqrt{\det(I_{3\nu} - \mho \Re L)}}
\end{align}
(similarly to (\ref{gintYY}) and (\ref{EYYupper})), where
\begin{equation}
\label{mho}
    \mho
    :=
    \bEc(\omega \omega^\rT)
    =
    2
    \blockdiag_{1\< k\< \nu}
    (\alpha_k, \beta_k, \alpha_k)
\end{equation}
is the diagonal covariance matrix of the auxiliary classical zero-mean Gaussian random vector $\omega$ with the variances (\ref{vars1}), (\ref{vars2}). The expectations in (\ref{Epsiprod2}) exist if the quantum and classical  covariance matrices $L$ and $\mho$ in (\ref{covcZ}) and (\ref{mho}) satisfy
\begin{equation}
\label{mhoReL}
    \br(\mho \Re L) < 1,
\end{equation}
where $\br(\cdot)$ is the spectral radius of a square matrix. The above calculations differ from the approach of \cite[Section~3]{PS_2015} involving a related problem of computing the MGF for the number operators  in QHOs and are summarised in the following theorem along with an additional insight into (\ref{mhoReL}).

\begin{theorem}
\label{th:QEF}
Suppose the vector $X$ of self-adjoint quantum variables $X_1, \ldots, X_n$ in (\ref{X}), satisfying the Weyl CCRs (\ref{CCR}), has a nonsingular CCR matrix $\Theta$ and is in a zero-mean Gaussian state with the quantum covariance matrix (\ref{covX}). Also, let the matrix $\Pi$ in (\ref{Xi}) be positive definite. Furthermore, suppose (\ref{mhoReL}) is satisfied for the matrices $L$ and $\mho$ in (\ref{covcZ}) and (\ref{mho}), associated with the transformation of $\Theta$ to a canonical form and symplectic diagonalization of $T^\rT \Pi T$ as described by (\ref{TJT})--(\ref{cZ}). Then the QEF $\Xi$ can be computed as in (\ref{Epsiprod2}). Moreover, for the fulfillment of (\ref{mhoReL}), it is sufficient  that
\begin{equation}
\label{Psmall}
  \br(P\Pi)
  \max_{1\<k \< \nu}
    \frac{\sinh(2\lambda_k)}{\lambda_k}
    < 1.
\end{equation}
\end{theorem}
\begin{proof}
With (\ref{Epsiprod2}) established above, it remains to discuss the condition (\ref{mhoReL}).
Since $\mho$ and $L$ are covariance matrices, (\ref{mhoReL}) is equivalent to the matrix $\mho \Re L$ being a contraction, that is, $\sqrt{\mho} \Re L \sqrt{\mho}\prec I_{3\nu}$. The latter inequality can be represented as
\begin{align}
\nonumber
    \sqrt{\mho}  &
    FV^{-1}T^{-1}P T^{-\rT}V^{-\rT}F^\rT
    \sqrt{\mho}\\
\nonumber
    =&
        \sqrt{\mho}
    FV^{-1}T^{-1}\Pi^{-1/2}\sqrt{\Pi}P\sqrt{\Pi}\Pi^{-1/2}\\
\label{prec}
     & \x T^{-\rT}V^{-\rT}F^\rT
    \sqrt{\mho}
    \prec I_{3\nu}
\end{align}
in view of (\ref{covcZ}).
Now, it follows from (\ref{TPT}) that
$$
    V^{-1}T^{-1} \Pi^{-1} T^{-\rT} V^{-\rT}
    =
    \Lambda^{-1} \ox I_2
    =
    \diag_{1\< k\< \nu}(\lambda_k^{-1}) \ox I_2
$$
and hence,
\begin{align}
\nonumber
    \sqrt{\mho}
    F&V^{-1}T^{-1} \Pi^{-1}T^{-\rT}V^{-\rT}F^\rT
    \sqrt{\mho}\\
\nonumber
    =&
    2
    \blockdiag_{1\< k\< \nu}
    \Big(
    \frac{1}{\lambda_k}
    \diag(\sqrt{\alpha_k}, \sqrt{\beta_k}, \sqrt{\alpha_k})\\
\nonumber
    & \x {\small\begin{bmatrix}
      1 & 0\\
      0 & 1\\
      1 & 0
    \end{bmatrix}
    \begin{bmatrix}
      1 & 0 & 1\\
      0 & 1 & 0
    \end{bmatrix}}
    \diag(\sqrt{\alpha_k}, \sqrt{\beta_k}, \sqrt{\alpha_k})
    \Big)\\
\nonumber
    =&
    2
    \blockdiag_{1\< k\< \nu}
    \left(
    \frac{1}{\lambda_k}
    \begin{bmatrix}
      \alpha_k & 0 & \alpha_k\\
      0 & \beta_k & 0\\
      \alpha_k & 0 & \alpha_k
    \end{bmatrix}
    \right)\\
\label{TPT1}
    \preccurlyeq &
    2\max_{1\<k \< \nu}
    \frac{\beta_k}{\lambda_k}
    I_{3\nu}
\end{align}
due to the structure of the matrices $F$ and $\mho$ in (\ref{cZ}) and (\ref{mho}). Here, use is made of the properties that the spectrum of the matrix
$
    {\small\begin{bmatrix}
      \alpha & 0 & \alpha\\
      0 & \beta & 0\\
      \alpha & 0 & \alpha
    \end{bmatrix}    }
$
is $\{0,2\alpha, \beta\}$ and that
$    2\alpha_k = \tanh\lambda_k< \lambda_k < \beta_k
$
in view of (\ref{sympfactk}) and $\lambda_k>0$ for all $k=1, \ldots, \nu$. Therefore, since $\sqrt{\Pi} P \sqrt{\Pi} \preccurlyeq \br(P\Pi) I_n$, then (\ref{TPT1}) implies that (\ref{Psmall}) is sufficient for (\ref{prec}), that is, (\ref{mhoReL}).
\end{proof}

The proof of Theorem~\ref{th:QEF} shows that  the condition (\ref{Psmall}), which is only sufficient in general, is also necessary for (\ref{mhoReL}) in the case when $P = \lambda \Pi^{-1}$ for some $\lambda \>0$. Also, being concerned with the quantum setting, (\ref{Psmall})   is more conservative than its counterpart $P\prec (2\Pi)^{-1}$ which secures the finiteness of the QEF (\ref{Xi}) in the classical Gaussian case. This conservativeness is particularly pronounced when at least one of the symplectic eigenvalues $\lambda_1, \ldots, \lambda_\nu$ of the matrix $T^\rT \Pi T$ in (\ref{TPT})  is large. In the classical limit, as the CCR matrix $\Theta$ goes to $0$, so does the matrix $T$ in (\ref{TJT}), and the symplectic spectrum of $T^\rT \Pi T$ collapses to $0$, in which case the maximum in (\ref{Psmall})  tends to $2$, thus recovering asymptotically the classical setting:
\begin{equation}
\label{Xiclass}
  \lim_{\Theta\to 0}
  \Xi =
  \frac{1}{\sqrt{\det(I_n - 2P\Pi})}.
\end{equation}
In the context of  quantum risk-sensitive  control \cite{J_2004,J_2005,VPJ_2018a,VPJ_2018b}, the weighting matrix $\Pi$ is usually scaled as $\Pi\mapsto \theta \Pi$ by a positive parameter  $\theta>0$, and so also is the symplectic spectrum of $\Pi$.   In this case, the QEF (\ref{Xi}) acquires a dependence on $\theta$, which is different from that of the corresponding classical cost $\frac{1}{\sqrt{\det(I_n - 2\theta P\Pi})}$ in (\ref{Xiclass}).

\section{Recurrence relations for 
quadratic-exponential functions of quantum processes}
\label{sec:two}

Complementing the previous discussion, consider a dis\-crete-time quantum process $X:=(X_0,X_1, X_2, \ldots)$, where $X_k$ is a vector consisting of an even number $n$ of self-adjoint quantum variables on an underlying Hilbert  space $\fH$ satisfying  the CCRs
\begin{equation}
\label{cXCCR}
    [\cX_N,  \cX_N^\rT] = 2i\Theta_N,
    \qquad
    \cX_N
    :=
    \begin{bmatrix}
        X_0\\
        \vdots\\
        X_N
    \end{bmatrix}
    =
    \begin{bmatrix}
        \cX_{N-1}\\
        X_N
    \end{bmatrix},
\end{equation}
where the CCR matrix $\Theta_N$ of order $(N+1)n$ admits the recurrence representation
\begin{equation}
\label{Thetanext}
    \Theta_N
    =
    \begin{bmatrix}
        \Theta_{N-1} & -\sigma_N^\rT\\
        \sigma_N & \theta_N
    \end{bmatrix},
    \qquad
    N=1,2,3,\ldots,
\end{equation}
in terms of auxiliary matrices $\sigma_N \in \mR^{n\x Nn}$ and $\theta_N$ (with the latter being a real antisymmetric matrix of order $n$)  such that $$
    [X_N,\cX_{N-1}^\rT] = 2i\sigma_N,
    \qquad
    [X_N,  X_N^\rT] = 2i\theta_N.
$$
In particular, the sequence $X$ can result as $    X_k = x(k \eps)
$
from considering a continuous-time quantum process $x$ at discrete instants with time step
$\eps>0$.  If $x$ is associated with the Heisenberg evolution of the system variables of an OQHO (which preserves the CCRs in time), then the matrix $\bTheta_N$ is block Toeplitz. However, the subsequent discussion does not employ this particular structure (and rather interprets $X$ as a general collection of indexed quantum variables).
We will be concerned with two classes of nonanticipating quadratic-exponential fun\-ctions of the quantum process $X$.
For a given but otherwise arbitrary time horizon $N = 0,1,2,\ldots$, consider a self-adjoint quantum variable
\begin{align}
\nonumber
    Q_N
    & :=
    R_N
    \re^{X_0^\rT C_0 X_0}
    R_N^\dagger\\
\label{QN}
    & =
    \re^{\cX_N^\rT C_N \cX_N}
    Q_{N-1}
    \re^{\cX_N^\rT C_N \cX_N},
\end{align}
where
\begin{equation}
\label{RN}
    R_N
    :=
    \lprod_{k=1}^N
    \re^{\cX_k^\rT C_k \cX_k}
    =
    \re^{\cX_N^\rT C_N \cX_N}
    R_{N-1}
\end{equation}
is an auxiliary (not necessarily self-adjoint) quantum variable, with
$C_0, C_1, C_2, \ldots$ being appropriately dimensioned real symmetric matrices. The initial conditions for (\ref{QN}) and (\ref{RN}) are
\begin{equation}
\label{Q0R0}
    Q_0 = \re^{X_0^\rT C_0 X_0},
    \qquad
    R_0 = \cI.
\end{equation}
The definitions (\ref{QN})--(\ref{Q0R0}) describe a discrete-time analogue of the quantum risk-sensitive cost functionals \cite{J_2004,J_2005} (see also \cite{DDJW_2006,YB_2009}) in measurement-based quantum control and filtering problems (employing continuous ti\-me ordered exponentials).
Another  class  of nonanticipating quadratic-exponential functions of the quantum process $X$ is  provided by
\begin{equation}
\label{EN}
    E_N
    :=
    \re^{\cX_N^\rT \Pi_N \cX_N},
\end{equation}
with $\Pi_N$ a real symmetric matrix of order $(N+1)n$, and correspond to recent developments \cite{VPJ_2018a,VPJ_2018b} on quantum extensions  of the classical risk-sensitive control problems \cite{BV_1985,J_1973,W_1981}. Although $E_N$ does not have a multiplicative structure in contrast to $Q_N$ and $R_N$, there is a correspondence between these two  classes of quantum variables $Q_N$ and $E_N$. To this end, we will need an entire function
\begin{align}
\nonumber
    \Ups(z)
    & :=
    \int_0^1
    \re^{zv}
    \rd v
    =
    \left\{
    \begin{matrix}
      1 & {\rm for} & z = 0\\
      \frac{\re^z-1}{z} & {\rm for} & z\ne 0
    \end{matrix}
    \right.\\
\label{Ups}
    & =
    \sum_{k=0}^{+\infty}
    \frac{z^k}{(k+1)!}
\end{align}
of a complex variable. Its role is clarified by the identities
\begin{align}
\label{exp1}
    \exp
    \Big(
    \begin{bmatrix}
        \alpha & 0\\
        \beta & 0
    \end{bmatrix}
    \Big)
    & =
    \begin{bmatrix}
        \re^\alpha & 0\\
        \beta \Ups(\alpha) & I
    \end{bmatrix},\\
\label{exp2}
    \exp
    \Big(
    \begin{bmatrix}
        0 & \beta\\
        0 & \alpha
    \end{bmatrix}
    \Big)
    & =
    \begin{bmatrix}
        I & \beta \Ups(\alpha)\\
        0 & \re^\alpha
    \end{bmatrix},
\end{align}
which hold for appropriately dimensioned matrices $\alpha$ and $\beta$ and are particular cases of the representation \cite{H_2008} for the matrix exponentials of block triangular matrices, with the function $\Ups$ from (\ref{Ups}) being evaluated in (\ref{exp1}) and (\ref{exp2}) at a square matrix $\alpha$.

\begin{theorem}
\label{th:PiC}
Suppose $\Pi_N$ and $C_N$ are real symmetric matrices of order $(N+1)n$ which satisfy the recurrence equation
\begin{align}
\nonumber
    \re^{4i\Theta_N \Pi_N}
    =&
    \re^{4i\Theta_N C_N}\\
\nonumber
    & \x\begin{bmatrix}
        \exp(4i\Theta_{N-1} \Pi_{N-1}) & & 0\\
        4i\sigma_N \Pi_{N-1} \Ups(4i\Theta_{N-1} \Pi_{N-1}) & & I_n
    \end{bmatrix}\\
\label{Pinext}
    & \x\re^{4i\Theta_N C_N},
    \qquad
    N =1,2,3,\ldots,
\end{align}
with the initial condition $\Pi_0 = C_0$. Then the quantum variables $Q_N$ in (\ref{QN}) and $E_N$ in (\ref{EN}) coincide:  $Q_N = E_N$. 
\end{theorem}
\begin{proof}
Since $\cX_0 = X_0$ in view of (\ref{cXCCR}),  the fulfillment of $\Pi_0 = C_0$ implies $Q_0 = E_0$ due to (\ref{Q0R0}) and (\ref{EN}) for $N=0$.
From Lemmas~\ref{lem:Dynk} and \ref{lem:tri} of Appendix~\ref{sec:prod}, it follows that if
\begin{align}
\nonumber
    \re^{4i\Theta_N \Pi_N}
    = &
    \lprod_{k=1}^N
    \exp
    \Big(4i\Theta_N
    \begin{bmatrix}
        C_k & 0\\
        0 & 0
    \end{bmatrix}
    \Big)\\
\nonumber
    & \x \exp
    \Big(
        4i\Theta_N
    \begin{bmatrix}
        C_0 & 0\\
        0 & 0
    \end{bmatrix}
    \Big)\\
\nonumber
    & \x\rprod_{k=1}^N
    \exp
    \Big(4i\Theta_N
    \begin{bmatrix}
        C_k & 0\\
        0 & 0
    \end{bmatrix}
    \Big)\\
\label{link}
    = &
    \re^{4i\Theta_N C_N}
    \exp
    \Big(4i\Theta_N
    \begin{bmatrix}
        \Pi_{N-1} & 0\\
        0 & 0
    \end{bmatrix}
    \Big)
    \re^{4i\Theta_N C_N},
\end{align}
then $Q_N=E_N$.
Here, use is made of the relations
$$
    \cX_k^\rT C_k \cX_k
    =
    \cX_N^\rT
    \begin{bmatrix}
        C_k & 0\\
        0 & 0
    \end{bmatrix}
    \cX_N,
    \qquad
    k=0,\ldots, N-1,
$$
which employ the fact that $\cX_k$ is a subvector of
$$
    \cX_N
    =
    \begin{bmatrix}
        \cX_k\\
        X_{k+1}\\
        \vdots\\
        X_N
    \end{bmatrix}
$$
in (\ref{cXCCR}) for any such $k$. The block partitioning (\ref{Thetanext}) implies that
$$
    \Theta_N
    \begin{bmatrix}
        \Pi_{N-1} & 0\\
        0 & 0
    \end{bmatrix}
    =
    \begin{bmatrix}
        \Theta_{N-1} \Pi_{N-1} & 0\\
        \sigma_N \Pi_{N-1} & 0
    \end{bmatrix},
$$
and hence, application of the identity (\ref{exp1}) yields
\begin{align}
\nonumber
    \exp&
    \Big(4i\Theta_N
    \begin{bmatrix}
        \Pi_{N-1} & 0\\
        0 & 0
    \end{bmatrix}
    \Big)\\
\label{expsparse1}
    & =
    \begin{bmatrix}
        \exp(4i\Theta_{N-1} \Pi_{N-1}) & & 0\\
        4i\sigma_N \Pi_{N-1} \Ups(4i\Theta_{N-1} \Pi_{N-1}) & & I_n
    \end{bmatrix}.
\end{align}
The relation (\ref{Pinext}) can now be obtained by substituting (\ref{expsparse1}) into (\ref{link}).
\end{proof}

We will now apply Theorem~\ref{th:PiC} to the case when the quadratic forms $\cX_k^\rT C_k \cX_k$ in (\ref{QN}) and (\ref{RN}) depend on the current quantum variables:
\begin{equation}
\label{XCDX}
    \cX_k^\rT C_k \cX_k = X_k^\rT D_k X_k,
\end{equation}
where $D_k$ is a real symmetric matrix of order $n$, so that
\begin{equation}
\label{CD}
    C_k
    =
    \begin{bmatrix}
        0 & 0\\
        0 & D_k
    \end{bmatrix}
\end{equation}
for all $k=1,2,3,\ldots$, with $C_0 = D_0$. The relation (\ref{XCDX}) is secured by  (\ref{CD}) since, in view of (\ref{cXCCR}), the vector $X_k$ is an appropriate subvector of $\cX_k$. A combination of (\ref{Thetanext}) with (\ref{CD}) leads to
$$
    \Theta_N
    C_N
    =
    \begin{bmatrix}
        0 & -\sigma_N^\rT D_N\\
        0 & \theta_N D_N
    \end{bmatrix},
$$
and hence,
\begin{equation}
\label{expsparse2}
    \re^{4i\Theta_N C_N}
    =
    \begin{bmatrix}
        I_n & -4i\sigma_N^\rT D_N \Ups(4i\theta_N D_N)\\
        0 & \exp(4i\theta_N D_N)
    \end{bmatrix}
\end{equation}
in view of the identity (\ref{exp2}). Substitution of (\ref{expsparse2}) into (\ref{Pinext}) of Theorem~\ref{th:PiC} yields
\begin{align}
\nonumber
    \re^{4i\Theta_N \Pi_N}
    = &
    \begin{bmatrix}
        I_n & -4i\sigma_N^\rT D_N \Ups(4i\theta_N D_N)\\
        0 & \exp(4i\theta_N D_N)
    \end{bmatrix}\\
\nonumber
    & \x
    \begin{bmatrix}
        \exp(4i\Theta_{N-1} \Pi_{N-1}) & & 0\\
        4i\sigma_N \Pi_{N-1} \Ups(4i\Theta_{N-1} \Pi_{N-1}) & & I_n
    \end{bmatrix}\\
\label{Pinext1}
    & \x
    \begin{bmatrix}
        I_n & -4i\sigma_N^\rT D_N \Ups(4i\theta_N D_N)\\
        0 & \exp(4i\theta_N D_N)
    \end{bmatrix}.
\end{align}
Together with the averaging results of Section~\ref{sec:QEF},  the recurrence relations of Theorem~\ref{th:PiC} and (\ref{Pinext1}) between the two classes of quadratic-exponential functions of quantum variables admit continuous-time versions for quantum processes governed by linear QSDEs. Such extensions  will be developed elsewhere (some of their aspects were tackled in \cite{VPJ_2018a,VPJ_2018b}).

\section{Conclusion}\label{sec:conc}

We have applied parameter randomization and complex symplectic factorization techniques to computing quad\-ra\-tic-exponential moments of Gaussian quantum variables. One of the ingredients of this approach is double averaging over the quantum state and auxiliary independent classical random variables. Another property, which underlies the method, is an isomorphism between the Lie algebras of complex Hamiltonian matrices  and quadratic forms of the quantum variables with complex symmetric matrices. Due to Dynkin's representation for the products of exponentials, this has led to complex symplectic factorizations, which reduce the quad\-ra\-tic-exponential moments to averaging the QCFs or MGFs   of quantum variables over an auxiliary classical Gaussian distribution. We have discussed their closed-form calculation in terms of quadratic-exponential moments of classical Gaussian random variables, and outlined a recursive implementation of this approach.
The computation of quadratic-exponential moments for a finite number of quantum variables, discussed in the present paper, is amenable to continuous-time extensions
concerning filtering and control problems for linear quantum stochastic systems with risk-sensitive performance criteria.

\appendix
\section{Commutator for quadratic forms of quantum variables satisfying CCRs}
\label{sec:quadcom}
\renewcommand{\theequation}{A\arabic{equation}}
\setcounter{equation}{0}

The following lemma is similar to the commutation relations for bilinear functions of annihilation and creation operators (see, for example, \cite[Appendix B]{JK_1998} and \cite[Lemma 4.2]{PUJ_2012}), whose exponentials  are considered in Schwinger's theorems \cite{CG_2010}.

\begin{lemma}
\label{lem:xieta}
For the vector $X$ of self-adjoint quantum variables $X_1,\\ \ldots, X_n$ in (\ref{X}) satisfying the CCRs (\ref{Theta}), the commutator of their quad\-ra\-tic forms
\begin{equation}
\label{xieta}
    \xi:= X^\rT A X,
    \qquad
    \eta:= X^\rT B X,
\end{equation}
specified by complex symmetric matrices $A$ and $B$ of order $n$, is computed as
\begin{equation}
\label{XCX}
    [\xi, \eta]
    =
    X^{\rT} C X.
\end{equation}
Here,
\begin{equation}
\label{CCC}
    C
    :=
    4i(A\Theta B - B \Theta A)
\end{equation}
is also a complex symmetric matrix of order $n$.
\end{lemma}
\begin{proof}
In view of (\ref{xieta}), it follows from the antisymmetry  and derivation properties of the commutator that
\begin{align}
\nonumber
    [\xi, X ]
    & =
    -[X,X^{\rT}]A X
    +
    (X^{\rT} A [X, X^{\rT} ])^{\rT}\\
\nonumber
    & =
    -2i\Theta A X
    +
    2i
    (X^{\rT} A \Theta )^{\rT}\\
\label{xiXcomm}
    & =
    -4i\Theta A X,
\end{align}
where use is also made of the symmetry of $A$ and the antisymmetry of the CCR matrix $\Theta$. By a similar reasoning, (\ref{xieta}) and (\ref{xiXcomm}) imply that
\begin{align*}
    [\xi,\eta]
    & =
    [\xi,X]^{\rT} B X
    +
    X^{\rT} B [\xi,X ]\\
    & =
    -4i((\Theta A X)^{\rT} B X
    +
    X^{\rT} B \Theta A X)\\
    & =
    4iX^{\rT}(A \Theta B
    -
    B \Theta A) X,
\end{align*}
which establishes (\ref{XCX}) in view of (\ref{CCC}).
\end{proof}

If a quadratic form of the quantum variables $X_1, \ldots, X_n$, mentioned in Lemma~\ref{lem:xieta}, is specified by an arbitrary (not necessarily symmetric) matrix, then its antisymmetric part  gives rise to an additive constant which does not contribute to the commutator. Indeed, due to the CCRs (\ref{Theta}), for any antisymmetric matrix $C:=(c_{jk})_{1\< j,k\< n} = -C^\rT \in \mC^{n\x n}$,
\begin{align*}
    X^{\rT} C X
    & =
    \sum_{j,k=1}^n
    c_{jk}X_jX_k
    =
    \sum_{j,k=1}^n
    c_{kj}X_kX_j\\
    & =
    -\sum_{j,k=1}^n
    c_{jk}X_kX_j
    =
    -\sum_{j,k=1}^n
    c_{jk}(X_jX_k - [X_j,X_k])\\
    & =
    -X^{\rT} C X + 2i
    \sum_{j,k=1}^n
    c_{jk}\theta_{jk}
    =
    -X^{\rT} C X -
    2i\Tr (C\Theta)\\
    & =
    -i \Tr (C\Theta).
\end{align*}
This clarifies why both quadratic forms in (\ref{xieta}), whose commutator in (\ref{XCX}) is the subject of the lemma, are specified by symmetric matrices without loss of generality. Also, Lemma~\ref{lem:xieta} shows that the quadratic forms of a given set of quantum variables, satisfying the CCRs, form a Lie algebra.

\section{A product formula for quadratic exponential functions of quantum variables with CCRs}
\label{sec:prod}
\renewcommand{\theequation}{B\arabic{equation}}
\setcounter{equation}{0}

In view of 
Lemma~\ref{lem:xieta}, the Lie algebra of quadratic forms $X^\rT C X$ of self-adjoint  quantum variables $X_1, \ldots, X_n$ satisfying the CCRs (\ref{CCR}) with a nonsingular CCR matrix $\Theta$ is isomorphic to the Lie algebra
\begin{equation}
\label{fC}
    \fC := \{\Theta C:\ C = C^\rT \in \mC^{n\x n}\}
\end{equation}
 of complex Hamiltonian matrices 
 (which generate the group of complex symplectic matrices $S\in \mC^{n\x n}$ satisfying $S\Theta S^\rT = \Theta$, so that $\re^{zG}$ is symplectic for any $z\in \mC$ and $G \in \fC$).
Indeed, since (\ref{CCC}) is equivalent to
$$
    4i\Theta C
    =
    [4i\Theta A, 4i\Theta B]
$$
(in view of $\det \Theta \ne 0$),
this isomorphism is described by the correspondence
\begin{equation}
\label{iso}
    X^\rT C X\longleftrightarrow 4i \Theta C,
\end{equation}
which leads to the following product formula for the exponentials of the quadratics (cf. \cite[Comment~3]{G_1974}).
To this end, we will need Dynkin's representation \cite{D_1947}
\begin{align}
\nonumber
    \ln (\re^\xi \re^\eta)
    = &
    -
    \sum_{k=1}^{+\infty}
    \frac{(-1)^k}{k}
    (\re^\xi \re^\eta - \cI)^k\\
\nonumber
    = &
    -
    \sum_{k=1}^{+\infty}
    \frac{(-1)^k}{k}
    \sum_{a_1,b_1,\ldots,a_k,b_k\>0:\ a_j+b_j>0\ {\rm for\ all}\ j=1,\ldots, k}\\
\nonumber
    & \rprod_{j=1}^k
    \Big(
    \frac{1}{a_j!b_j!}\xi^{a_j}\eta^{b_j}
    \Big)\\
\nonumber
    =&
    -
    \sum_{k=1}^{+\infty}
    \frac{(-1)^k}{k}
    \sum_{a_1,b_1,\ldots,a_k,b_k\>0:\ a_j+b_j>0\ {\rm for\ all}\ j=1,\ldots, k}\\
\nonumber
    & \frac{1}{\sum_{j=1}^k (a_j+b_j)}
    \Big(\rprod_{j=1}^{k-1}
    \ad_\xi^{a_j}
    \ad_\eta^{b_j}
    \Big)
    ((
    \ad_\xi^{a_k}
    \ad_\eta^{b_k-1})(\eta))\\
\label{Dynk}
    =&
    \xi + \eta + T(\ad_\xi,\ad_\eta,[\xi,\eta]),
\end{align}
where $\xi$ and $\eta$ are (in general, noncommuting) variables,
and the linear superoperator $\ad_\sigma:= [\sigma,\cdot]$ is the commutator with a given operator $\sigma$. Here,
\begin{align}
\nonumber
    T(\alpha,\beta,\gamma)
    =&
    \sum_{k=1}^{+\infty}
    \sum_{a_1,b_1,\ldots, a_k,b_k\>0}
    c_{a_1b_1\ldots a_kb_k}\\
\label{DynkT}
    & \x
    (\alpha^{a_1}\circ \beta^{b_1}\circ \ldots \circ \alpha^{a_k}\circ \beta^{b_k})(\gamma)
\end{align}
is a formal power series which consists of monomials of superoperators $\alpha$ and $\beta$ acting on an operator $\gamma$, where the coefficients
$c_{a_1b_1\ldots a_kb_k} \in \mR$ can be recovered from (\ref{Dynk}). For what follows, we associate with the nonsingular CCR matrix $\Theta$ a map $\cE$ which maps a complex symmetric matrix $C=C^\rT\in \mC^{n\x n}$ to the matrix
\begin{equation}
\label{cE}
    \cE(C):= \re^{4i\Theta C}.
\end{equation}
This can also be identified with the exponential map of the infinitesimal generators from the Lie algebra $\fC$ in (\ref{fC}) to the group of complex symplectic matrices.

\begin{lemma}
\label{lem:Dynk}
Suppose the CCR matrix $\Theta$ in (\ref{CCR}) is nonsingular, and $\xi$ and $\eta$  are the quantum variables described in Lemma~\ref{lem:xieta}. Then
\begin{equation}
\label{XEX}
    \re^{\xi} \re^{\eta} = \re^{X^\rT E X},
\end{equation}
where $E$ is a complex symmetric matrix of order $n$ computed as
\begin{equation}
\label{EAB}
    E = (4i\Theta )^{-1} \ln(\cE(A)\cE(B)),
\end{equation}
with $\cE$ the map from (\ref{cE}).
\end{lemma}
\begin{proof}
By applying (\ref{Dynk}) and (\ref{DynkT}) to the quantum variables in (\ref{xieta}),
specified by arbitrary symmetric matrices $A, B \in \mC^{n\x n}$,  and using the Lie-algebraic isomorphism (\ref{iso}), it follows that
$$
    \zeta := \ln (\re^\xi \re^\eta)
    =
    X^\rT E X,
$$
where the matrix $4i\Theta E$ is related to $4i\Theta A$ and $4i\Theta B$ in the same way as the quantum variable $\zeta$ is related to $\xi$ and $\eta$:
$$
    4i\Theta E = 4i\Theta A+4i\Theta B + T(\ad_{4i\Theta A}, \ad_{4i\Theta B}, [4i\Theta A, 4i\Theta B]),
$$
and hence,
$$
    4i\Theta E = \ln(\cE(A)\cE(B)),
$$
thus establishing (\ref{XEX}) and (\ref{EAB}).
\end{proof}

By induction, Lemma~\ref{lem:Dynk} can be extended to products of an arbitrary number of quadratic-exponential functions of the quantum variables:
\begin{equation}
\label{CCE}
    \rprod_{k=1}^N
    \re^{X^\rT C_k X} = \re^{X^\rT E X},
    \qquad
    E= (4i\Theta)^{-1}\ln \rprod_{k=1}^N \cE(C_k),
\end{equation}
where $C_1, \ldots, C_N$ are complex symmetric matrices.  For example, if $N=3$ and $C_3:=-C_1$, the matrix $E$ in (\ref{CCE}) admits a closed-form representation:
\begin{align}
\nonumber
    E
    & =
    (4i\Theta)^{-1}
    \ln
    (
        \cE(C_1)
        \cE(C_2)
        \cE(-C_1)
    )\\
\nonumber
    & =
    (4i\Theta)^{-1}
    \ln
    \exp
    (4i
        \cE(C_1)
        \Theta C_2
        \cE(-C_1)
    )\\
\nonumber
    & =
    \Theta^{-1}
        \cE(C_1)
        \Theta C_2
        \cE(-C_1)\\
\nonumber
    & =
        \cE(\Theta^{-1}C_1\Theta )
        C_2
        \cE(-C_1)\\
\nonumber
        & =
        \re^{4iC_1 \Theta }
        C_2
        \re^{-4i\Theta C_1},
\end{align}
where we have repeatedly used the identity $Sf(Z)S^{-1} = f(SZ S^{-1})$ for similarity transformations of functions of matrices \cite{H_2008}.
 Of particular interest for our purposes is the following case, where the presence of symmetries  leads to additional properties of the resulting matrix $E$ in (\ref{CCE}).

\begin{lemma}
\label{lem:tri}
Suppose the self-adjoint quantum variables $X_1, \ldots, X_n$ in (\ref{X}) have a nonsingular CCR matrix $\Theta$ in (\ref{CCR}). Also, let $C_1$ and $C_2$ be real symmetric matrices of order $n$. Then the relation (\ref{CCE}), considered for $N=3$ and $C_3:= C_1$, holds with a real symmetric matrix $E$.
\end{lemma}
\begin{proof}
Since the matrix $E$ is symmetric, it remains to prove that $E$ is real. In the case under consideration, the second equality in (\ref{CCE}) takes the form
$$
    E =
    (4i\Theta)^{-1}
    \ln
    (
    \cE(C_1)
    \cE(C_2)
    \cE(C_1)
    ),
$$
whose complex conjugation leads to
\begin{align*}
    \overline{E}
    & =
    -(4i\Theta)^{-1}
    \ln
    \overline{
    \big(
    \cE(C_1)
    \cE(C_2)
    \cE(C_1)
    \big)}\\
    & =
    -(4i\Theta)^{-1}
    \ln
    \big(
    \cE(-C_1)
    \cE(-C_2)
    \cE(-C_1)
    \big)\\
    & =
    -(4i\Theta)^{-1}
    \ln
    \Big(
    \big(
    \cE(C_1)
    \cE(C_2)
    \cE(C_1)
    \big)^{-1}
    \Big)\\
    & =
    (4i\Theta)^{-1}
    \ln
    (
    \cE(C_1)
    \cE(C_2)
    \cE(C_1)
    ) = E,
\end{align*}
thus implying that $E \in \mR^{n\x n}$. Here, we have also used the properties $\overline{\ln C} = \ln \overline{C}$ and $\ln (C^{-1}) = -\ln C$ of the logarithm for complex matrices.
\end{proof}

The properties of the matrix $E$ in Lemma~\ref{lem:tri} are in line with the self-adjointness of the quantum variable
$
    \re^{X^\rT C_1 X}
    \re^{X^\rT C_2 X}
    \re^{X^\rT C_1 X}
$ in the case of real symmetric matrices $C_1$ and $C_2$. Again, by induction, the lemma extends as
\begin{equation}
\label{CCEsymm1}
    \Big(\lprod_{k=1}^N
    \re^{X^\rT C_k X}
    \Big)
    \re^{X^\rT C_0 X}
    \rprod_{k=1}^N
    \re^{X^\rT C_k X}
    =
    \re^{X^\rT E X},
\end{equation}
with
\begin{equation}
\label{CCEsymm2}
    \cE(E)
    =
    \Big(
    \lprod_{k=1}^N
    \cE(C_k)
    \Big)
    \cE(C_0)
    \rprod_{k=1}^N
    \cE(C_k)
\end{equation}
for arbitrary real symmetric matrices $C_0, \ldots, C_N$ leading to a real symmetric matrix $E$.

\section{Complex symplectic factorizations for second order matrices}
\label{sec:fact}
\renewcommand{\theequation}{C\arabic{equation}}
\setcounter{equation}{0}

Consider an inverse problem of finding a factorization (\ref{CCEsymm2}) for the matrix  $\cE(E)$, given by (\ref{cE}), in the case of two quantum variables ($n=2$) for a given diagonal matrix
\begin{equation}
\label{EEE}
    E :=
    \begin{bmatrix}
      a & 0\\
      0 & b
    \end{bmatrix}
    =
    a E_1 + bE_2,
\end{equation}
where $a$ and $ b$  are positive real numbers, and
\begin{equation}
\label{E12}
      \qquad
    E_1 :=
    \begin{bmatrix}
        1 & 0 \\
        0 & 0
    \end{bmatrix},
    \qquad
    E_2 :=
    \begin{bmatrix}
        0 & 0 \\
        0 & 1
    \end{bmatrix}
\end{equation}
are auxiliary matrices.
 In the case of $(2\x 2)$-matrices being considered, the CCR matrix can be represented as
\begin{equation}
\label{cbJ}
    \Theta
    =
    \theta \bJ
    =
    \begin{bmatrix}
        0 & \theta\\
        -\theta & 0
    \end{bmatrix},
\end{equation}
where $\theta\in \mR\setminus \{0\}$, and  use is made of the matrix $\bJ$ from  (\ref{bJ}). Since
$$
    (i \Theta E )^2
    =
    -\theta ^2
    \begin{bmatrix}
      0 & b\\
      -a & 0
    \end{bmatrix}^2
    =
    ab\theta ^2
    I_2
$$
in view of (\ref{EEE})--(\ref{cbJ}), then
the corresponding map (\ref{cE}) acts as
\begin{align}
\nonumber
    \cE(E)
    & =
    \sum_{k\>0}
    (4i\Theta E)^{2k}
    \Big(
        \frac{1}{(2k)!}
        I_2
        +
        \frac{1}{(2k+1)!}
        4i\Theta E
    \Big)\\
\nonumber
    & =
    \sum_{k\>0}
    (16ab\theta ^2)^k
    \left(
        \frac{1}{(2k)!}
            I_2
        +
        \frac{4\theta i}{(2k+1)!}
        \begin{bmatrix}
          0 & b\\
          -a & 0
        \end{bmatrix}
    \right)\\
\label{expTE}
    & =
    \cosh(4\theta \sqrt{ab}) I_2  +
    i
    \sinh(4\theta \sqrt{ab})
    \begin{bmatrix}
        0 &   \sqrt{\tfrac{b}{a}}\\
        -\sqrt{\tfrac{a}{b}} & 0
    \end{bmatrix}.
\end{align}
In the limit, as $b\to 0$ or $a\to 0$, the identity (\ref{expTE}) leads respectively to
\begin{equation}
\label{lims}
    \cE(aE_1)
    =
    \begin{bmatrix}
      1 & 0\\
      -4a\theta i & 1
    \end{bmatrix},
    \qquad
    \cE(bE_2)
    =
    \begin{bmatrix}
      1 & 4b\theta i\\
      0 & 1
    \end{bmatrix},
\end{equation}
which can also be obtained by using the nilpotence of strictly triangular matrices (more precisely, the second powers of the matrices
$
    \bJ E_1
    =
    {\small\begin{bmatrix}
      0 & 0\\
      -1 & 0
    \end{bmatrix}}
$ and $
    \bJ E_2
    =
    {\small\begin{bmatrix}
      0 & 1\\
      0 & 0
    \end{bmatrix}}
$ vanish). The triangular matrices (\ref{lims}) provide ``building blocks'' for the following  complex symplectic factorization.

\begin{lemma}
\label{lem:sympfact}
For any $a,b>0$ in (\ref{EEE}), there exist
$\alpha, \beta \in \mR$ satisfying
\begin{equation}
\label{sympfact}
    \cE(\alpha E_1)
    \cE(\beta E_2)
    \cE(\alpha E_1)
    =
    \cE(E),
\end{equation}
where the map $\cE$ in (\ref{cE}) is associated with the CCR matrix $\Theta$ in (\ref{cbJ}): \begin{equation}
\label{sympfact1}
    \alpha
    =
    \frac{1}{4\theta}
    \tanh(2\theta \sqrt{ab})
    \sqrt{\frac{a}{b}},
    \quad
    \beta
    =
    \frac{1}{4\theta}
    \sinh(4\theta \sqrt{ab})
    \sqrt{\frac{b}{a}}.
\end{equation}
\end{lemma}
\begin{proof}
Indeed, the left-hand side of (\ref{sympfact}) is computed as
\begin{align}
\nonumber
\cE(\alpha E_1) &\cE(\beta E_2) \cE(\alpha E_1)\\
\nonumber
 & =
    \begin{bmatrix}
      1 & 0\\
      -4\alpha \theta i & 1
    \end{bmatrix}
    \begin{bmatrix}
      1 & 4\beta \theta i\\
      0 & 1
    \end{bmatrix}
    \begin{bmatrix}
      1 & 0\\
      -4\alpha \theta i & 1
    \end{bmatrix}\\
\nonumber
& =
    \begin{bmatrix}
      1 & 0\\
      -4\alpha \theta i & 1
    \end{bmatrix}
    \begin{bmatrix}
      1+16 \alpha \beta \theta^2 & 4\beta \theta i\\
      -4\alpha \theta i  & 1
    \end{bmatrix}\\
\label{left}
& =
    \begin{bmatrix}
      1+16 \alpha \beta \theta^2 & 4\beta \theta i\\
      -8\alpha \theta(1+8 \alpha \beta \theta^2) i  & 1+16 \alpha \beta \theta^2
    \end{bmatrix},
\end{align}
which is a matrix with unit determinant. By comparing the right-hand sides  of (\ref{expTE}) and (\ref{left}), it follows that (\ref{sympfact}) is equivalent to a set of two equations
$$
    1+16 \alpha \beta \theta^2
    =
    \cosh(4\theta \sqrt{ab}),
    \qquad
    4\beta \theta
     =
    \sinh(4\theta \sqrt{ab})
    \sqrt{\frac{b}{a}}
$$
for $\alpha$ and $\beta$,
whose solution is given by (\ref{sympfact1}) in view of the identities $\cosh (2\phi) - 1= 2(\sinh \phi)^2$ and $\sinh (2\phi) = 2\sinh \phi \cosh \phi$.
\end{proof}

In particular, for the quantum-mechanical position and momentum operators $q$ and $p:= -i\d_q$, whose CCR matrix is given by (\ref{cbJ}) with $\theta= \frac{1}{2}$ in view of $[q,p]=i\cI$, application of (\ref{CCEsymm1}), (\ref{CCEsymm2}), (\ref{sympfact}) and (\ref{sympfact1}) leads to
\begin{align}
\nonumber
    \re^{a q^2 + b p^2}
    & =
    \re^{a q^2 - b \d_q^2}\\
\label{eee}
    & =
    \re^{\alpha q^2}
    \re^{\beta p^2}
    \re^{\alpha q^2}
    =
    \re^{\alpha q^2}
    \re^{-\beta \d_q^2}
    \re^{\alpha q^2},
\end{align}
where
\begin{equation}
\label{sympfact3}
    \alpha
    =
    \frac{1}{2}
    \tanh(\sqrt{ab})
    \sqrt{\frac{a}{b}},
    \qquad
   \beta
     =
    \frac{1}{2}
    \sinh(2\sqrt{ab})
    \sqrt{\frac{b}{a}}.
\end{equation}
There is an analogy between (\ref{sympfact}) (and its application to (\ref{eee})) and the factorizations in \cite[Eqs. (5), (6)]{CG_2010}.
The relations (\ref{eee}) (with (\ref{sympfact3})) can be verified
on a dense set of infinitely differentiable functions $f:\mR\to \mC$ (of the position variable) with a sufficiently fast decay of derivatives at infinity (a subset of the Schwartz space \cite{V_2002}). The following lemma outlines such a verification for Gauss\-i\-an-shaped test functions
\begin{equation}
\label{fgauss}
    f(q)
    :=
    \re^{-\gamma(q-\mu)^2},
    \qquad
    \gamma > 0,\
    \mu \in \mR.
\end{equation}
Up to a normalization factor, (\ref{fgauss}) describes the Gaussian PDF $\sqrt{\frac{\gamma}{\pi}}f$ with mean $\mu$ and variance $\frac{1}{2\gamma}$.

\begin{lemma}
\label{lem:ver}
The relation (\ref{eee}) for the position and momentum operators, with $\alpha$ and $\beta$ given by (\ref{sympfact3}), holds at the test functions $f$ in (\ref{fgauss}), with \begin{equation}
\label{gammagood}
    \frac{\alpha}{1+4\alpha \beta} < \gamma-\alpha < \frac{1}{4\beta}.
\end{equation}
\end{lemma}
\begin{proof}
Since $\gamma>\alpha$, then
\begin{align}
\nonumber
    \re^{\alpha q^2}
    f(q)
    & =
    \re^{\alpha q^2 -\gamma(q-\mu)^2}\\
\nonumber
    & =
    \re^{-(\gamma-\alpha)\big(q-\tfrac{\gamma\mu}{\gamma-\alpha}\big)^2 - \gamma \mu^2 + \tfrac{(\gamma \mu)^2}{\gamma-\alpha}}\\
\label{opf1}
    &=
    \re^{-(\gamma-\alpha)\big(q-\tfrac{\gamma\mu}{\gamma-\alpha}\big)^2 +
    \tfrac{\alpha \gamma \mu^2}{\gamma-\alpha}}
\end{align}
is proportional to the Gaussian PDF with mean $\frac{\gamma\mu}{\gamma-\alpha}$ and variance $\frac{1}{2(\gamma-\alpha)}$. Now, $\re^{\beta \d_q^2}$ is an integral operator with the Markov transition kernel (at time $t = 2\beta$) for an auxiliary Brownian motion $\xi(t)= \xi(0) + W(t)$, where $W$ is the standard Wiener process \cite{KS_1991} independent of the initial condition $\xi(0)$. For any $t>0$, the PDF $\phi(t,q)$ of $\xi(t)$ satisfies the heat PDE \cite{E_1998} $\d_t \phi = \frac{1}{2}\d_q^2 \phi$, so that $\phi(2\beta,q) = \re^{\beta \d_q^2} \phi(0,q)$.
With the class of Gaussian PDFs being invariant under the heat  semigroup,
\begin{equation}
\label{opf2}
    \re^{-\beta \d_q^2}
    \Big(
    \sqrt{\frac{\gamma}{\pi}}
    f(q)
    \Big)
    =
    \sqrt{\frac{\wt{\gamma}}{\pi}}
    \re^{-\wt{\gamma}(q-\mu)^2},
    \qquad
    \wt{\gamma}
    :=
    \frac{\gamma}{1-4\beta \gamma},
\end{equation}
where the mean value $\mu$ remains unchanged, and  $\wt{\gamma}$
is found from the equation $\frac{1}{2\wt{\gamma}} + 2\beta= \frac{1}{2\gamma}$ which relates the variances  of the Gaussian PDFs under the action of the heat semigroup, provided $\gamma < \frac{1}{4\beta}$. Under the action of the right-hand side of (\ref{eee}) on (\ref{fgauss}), the  parameter $\gamma$ is transformed as
$$
    \gamma\mapsto \gamma - \alpha \mapsto \frac{\gamma-\alpha}{1-4\beta (\gamma-\alpha)}\mapsto \frac{\gamma-\alpha}{1-4\beta (\gamma-\alpha)}-\alpha,
$$
so that the corresponding  variances are all positive if and only if $\gamma$ satisfies (\ref{gammagood}).
Now,  consider the Gaussian-shaped functions under the action of the left-hand side of (\ref{eee}). To this end, we note that if the parameters $\gamma$, $\mu$ in (\ref{fgauss}) and an additional quantity $\sigma>0$ are smooth functions of time $t\>0$,  then $\psi(t,q):= \sigma f(q)$ satisfies
\begin{equation}
\label{fdot}
    \d_t \psi
    =
    (
    -\dot{\gamma} q^2
    + 2(\gamma \mu)^{^\centerdot} q
    +
    (\ln \sigma - \gamma \mu^2)^{^\centerdot})
    \psi,
\end{equation}
where $(\cdot)^{^\centerdot}:=- \d_t (\cdot)$ denotes the partial time derivative. Furthermore,
\begin{equation}
\label{f''}
  f''(q)
  =
  2\gamma (2\gamma q^2 -4\gamma \mu q +2\gamma\mu^2-1)f(q).
\end{equation}
From (\ref{fdot}) and (\ref{f''}), it follows that the function $\psi$ satisfies the PDE
\begin{equation}
\label{psiPDE}
    \d_t \psi
    =
(a q^2 - b \d_q^2) \psi
\end{equation}
if the functions $\gamma$, $\mu$ and $\sigma$ are governed by the set of three ODEs
\begin{align}
\label{dot1}
    \dot{\gamma}
    & =
    4b \gamma^2 - a,\\
\label{dot2}
    (\gamma \mu)^{^\centerdot}
    & =
    4b\gamma^2 \mu,\\
\label{dot3}
    (\ln\sigma - \gamma \mu^2)^{^\centerdot}
    & =
    2b\gamma (1-2\gamma \mu^2),
\end{align}
the first of which is a scalar Riccati equation. The solution of (\ref{dot1})--(\ref{dot3}) leads to a manifold of Gaussian-shaped solutions of the PDE (\ref{psiPDE}). Its comparison (at time $t=1$)  with the action of the operator on the right-hand side of (\ref{eee}) on (\ref{fgauss}), that is, $    (\re^{\alpha q^2}
    \re^{-\beta \d_q^2}
    \re^{\alpha q^2})(f)$,  computed using
 (\ref{opf1}) and (\ref{opf2}), completes the proof. The details of these calculations are tedious and are omitted.
\end{proof}
\end{document}